\documentclass[10pt,A4]{article}

\usepackage{bbm}

\usepackage{mdframed}
\usepackage{hyperref}

\usepackage[margin=2cm]{geometry}

\usepackage{amsfonts,amsthm,amssymb}
\usepackage{amsmath}

\begin{document}

\setcounter{page}{0}

\title{Approximate  Aggregate Queries Under
Additive Inequalities}

\author{Mahmoud Abo-Khamis \\ Relational AI \\ mahmoud.abokhamis@relational.ai \and Sungjin Im~\thanks{Supported in part by NSF grants CCF-1409130, CCF-1617653, and CCF-1844939.} \\ University of California, Merced \\ sim3@ucmerced.edu \and Benjamin Moseley~\thanks{Supported in part by NSF grants CCF-1725543, 1733873, 1845146, a Google Research Award, a Bosch junior faculty chair and an Infor faculty award.} \\ Carnegie Mellon University \\ moseleyb@andrew.cmu.edu \and Kirk Pruhs~\thanks{Supported in part by NSF grants CCF-1421508 and CCF-1535755, and an IBM Faculty Award.} \\ University of Pittsburgh \\ kirk@cs.pitt.edu \and Alireza Samadian \\ University of Pittsburgh \\ samadian@cs.pitt.edu}

\newcommand{\functionname}[1]{\text{\sf #1}}
\newcommand{\fhtw}{\functionname{fhtw}}
\newcommand{\tw}{\functionname{tw}}
\newcommand{\atoms}{\text{atoms}}
\newcommand{\arity}{\text{arity}}
\newcommand{\vars}{\textnormal{vars}}

\newcommand{\calH}{\mathcal H}
\newcommand{\calV}{\mathcal V}
\newcommand{\calE}{\mathcal E}
\newcommand{\calD}{\mathcal D}

\newcommand{\D}{\mathbf{D}}

\newcommand{\be}{\begin{enumerate}}
\newcommand{\ee}{\end{enumerate}}
\newcommand{\bi}{\begin{itemize}}
\newcommand{\ei}{\end{itemize}}
\newcommand{\beq}{\begin{equation}}
\newcommand{\eeq}{\end{equation}}

\newcommand{\bp}{\begin{proof}}
\newcommand{\ep}{\end{proof}}
\newcommand{\bcor}{\begin{cor}}
\newcommand{\ecor}{\end{cor}}
\newcommand{\bthm}{\begin{thm}}
\newcommand{\ethm}{\end{thm}}
\newcommand{\blmm}{\begin{lmm}}
\newcommand{\elmm}{\end{lmm}}
\newcommand{\bdefn}{\begin{defn}}
\newcommand{\edefn}{\end{defn}}
\newcommand{\bprop}{\begin{prop}}
\newcommand{\eprop}{\end{prop}}
\newcommand{\bconj}{\begin{conj}}
\newcommand{\econj}{\end{conj}}
\newcommand{\bopm}{\begin{opm}}
\newcommand{\eopm}{\end{opm}}
\newcommand{\brmk}{\begin{rmk}}
\newcommand{\ermk}{\end{rmk}}

\newcommand{\suchthat}{\ | \ }

\theoremstyle{plain}                   
\newtheorem{theorem}{Theorem}
\newtheorem{lemma}[theorem]{Lemma}
\newtheorem{definition}[theorem]{Definition}
\newtheorem{corollary}[theorem]{Corollary}

\theoremstyle{definition}              
\newtheorem{pbm}{Problem}
\newtheorem{assumption}{Assumption}
\newtheorem{opm}{Open Problem}
\newtheorem{conj}{Conjecture}
\newtheorem{ex}{Example}
\newtheorem{exer}{Exercise}
\newtheorem{defn}{Definition}
\newtheorem{alg}{Algorithm}
\newtheorem{rmk}{Remark}
\newtheorem{claim}{Claim}
\newtheorem{fact}{Fact}

\maketitle

\begin{abstract}
We consider the problem of evaluating certain types of  functional aggregation queries  on relational data subject to additive inequalities.
Such aggregation queries, with a smallish number of additive inequalities, arise naturally/commonly in many applications, particularly in 
learning applications. 
We give a relatively complete categorization of the computational complexity of such problems. 
We first show that the problem is NP-hard, even in the case of one additive inequality.
Thus we turn to approximating the query. 
Our main result is an efficient algorithm for approximating, with arbitrarily small relative error,
many natural aggregation queries with one additive inequality.
We give examples of natural queries that can be efficiently solved using this algorithm.
In contrast we show that the situation with two additive inequalities is quite different, by showing
that it is NP-hard to evaluate simple aggregation queries, with two additive inequalities, with any bounded relative error.
\end{abstract}

\newpage

\section{Introduction}

Kaggle  surveys~\cite{KaggleSurvey} show that 
the majority of learning tasks faced by  data scientists involve relational data. Most commonly the relational data is stored in tables in a relational database. So these data scientists want to compute something like
\begin{quote}
\textbf {Data Science Query} =  Standard\_Learning\_Task(Relational Tables $T_1, \ldots T_m$)
\end{quote}
However, almost all standard algorithms for standard
learning problems 
assume that the input consists of points in Euclidean space~\cite{HundredPage},
and thus 
are not designed to  operate directly on relational data. 
The current standard practice for a data scientist, confronted with a learning task on relational data, is: 

\begin{enumerate}
\item
Firstly, convert any nonnumeric categorical data to numeric data.  As there are standard methods to accomplish this~\cite{HundredPage}, and
as we do not innovate with respect to this process, we will assume that all data is a priori numerical, so we need not  consider this step. 
\item
Secondly, issue a feature extraction query to extract the  data from the relational database by joining together the tables to materialize a design matrix $J=T_1 \Join \dots \Join T_m$ with say $N$ rows and $d$ columns/features.  Each row of this design matrix is then interpreted as a point in $d$-dimensional Euclidean space. 
\item
Finally this design matrix $J$ is important into a  standard learning algorithm to train the model.
\end{enumerate}
Thus conceptually, standard  practice transforms a  data science query to a query of the following form:
\begin{quote}
\textbf {Data Science Query} =  Standard\_Learning\_Algorithm(Design Matrix $J = T_1 \Join \dots \Join T_m$)
\end{quote}
where the joins are evaluated first, and the learning algorithm is then applied to the result. 
Note that if each table has $n$ rows, the design
matrix $J$ can have as many as $n^m$ entries. 
Thus, independent of the learning task, this standard practice necessarily has exponential worst-case time and space complexity as the design matrix can be exponentially larger than the underlying relational tables. 
Thus a natural research question is what we call the relational learning question:
\begin{quote}
\textbf {Relational Learning Question:}  What standard learning problems admit relational algorithms, which informally
are algorithms that are efficient when the data is in relational form? 
\end{quote}
Note that a relational algorithm can not afford to explicitly join the relational tables.

In this paper we consider the relational learning question in the context of 
the problem of evaluating Functional Aggregate Queries (FAQ's) subject to 
additive constraints, which we call
FAQ-AI queries. Such queries/problems, with a smallish number of inequalities, arise naturally as subproblems in many  standard learning algorithms. 
Before formally defining an FAQ-AI query, let us  start with some examples. The first collection of examples are  related to the classic Support Vector Machines problem (SVM),
in which points are classified
based on the side of a hyperplane that the point lies on.~\cite{HundredPage,Shalev-Shwartz} Each of the following examples can be reduced to  FAQ-AI queries
with one additive inequality:
\begin{itemize}
    \item Counting the number of points correctly (or incorrectly) classified by a hyperplane.
    \item Finding the minimum distance of a 
    correctly classified point to the boundary of a given hyperplane.
    \item
    Computing the gradient of the SVM objective function at a particular point. 
\end{itemize}
And now we give some examples of problems related to the classic $k$-means clustering problem~\cite{Shalev-Shwartz}, in which the goal is to find
locations for $k$ centers so as to minimize the aggregate 2-norm squared distance from
each point to its closest center. Each of the following
examples can be reduced to FAQ-AI queries with $k-1$ inequalities:
\begin{itemize}
\item
Evaluating the $k$-means objective value for a particular collection of $k$ centers.
    \item Computing the new centers in one iteration of the commonly used Lloyd's algorithm.
    \item
Computing the furthest point in each cluster from the center of that cluster. 
\end{itemize}
All of these problems are  readily solvable in nearly linear time in the size of the input 
if the input was the design matrix.
Our goal  to determine whether relational algorithms
exist for 
such FAQ-AI problems
when the input is in relational form. 

One immediate difficulty that we run into is that
if the tables have a sufficiently complicated structure, almost all natural problems/questions
about the design matrix are NP-hard if the data is in relational form. For example, it is NP-hard to even determine whether or not the design matrix is empty or not (see for example \cite{grohe2006structure,marx2013tractable}). 
Thus, as we want to focus on the complexity of the functional aggregate query and the additive
inequalities, we conceptually want to   abstract out the complexity 
of the tables. The simplest way to accomplish this is to  primarily focus on
 instances where the structure of the tables is
simple, with the most natural candidate for ``simplicity'' being that the join is acyclic. Acyclic joins are the norm in practice, and are a commonly considered special case in the database literature. 
For example,  there are efficient algorithms to compute the size of the design matrix for acyclic joins.  

Formally defining what an ``relational'' algorithm is
problematic, as for each  natural candidate definition there
are plausible scenarios in which that candidate definition is not the ``right'' definition. But for the purposes of this paper
it is sufficient to think of a ``relational'' algorithm as one whose
runtime is polynomially bounded in $n$, $m$ and $d$ if the join is acyclic.

\subsection{Formal Definitions}

Unfortunately, a formal definition of FAQ-AI is rather cumbersome,
and notation heavy. To aid the reader, after the formal definitions, we give some examples of how to model some of the
particular  learning problems discussed earlier, as FAQ-AI problems. Due to space limitations some definitions, that are more standard or less critical, and some proofs, have been moved to the appendix.

The input to FAQ-AI problem consists of three components:
\begin{itemize}
\item
A collection of relational tables $T_1, \ldots T_m$ with real-valued entries. Let
$J = T_1 \Join T_2 \Join \dots \Join T_m$ be the
design matrix that arises from the inner join of the tables. 
Let $n$ be an upper bound on the number of rows in any table $T_i$,
let $N$ be the number of rows in $J$, and let $d$ be the number of columns/features in $J$. 
\item
An FAQ query $Q(J)$ that is either a SumProd query or a SumSum query. 
We define a SumSum query to be a query of the form:
\begin{align*}
   Q(J) = \bigoplus_{x \in J} \bigoplus_{i=1}^d F_{i}(x_i)
\end{align*}
where $(R, \oplus, I_0)$ is a commutative monoid over the arbitrary set $R$ with identity $I_0$.
We define a SumProd query to be a query of the form: 
\begin{align*}
   Q(J) =  \bigoplus_{x \in J} \bigotimes_{i=1}^d F_{i}(x_i)
\end{align*}
where $(R, \oplus, \otimes, I_0, I_1)$ is a commutative semiring over the arbitrary set $R$ with additive identity $I_0$
and multiplicative identity $I_1$. 
In each case, $x$ is a row in the design matrix $J$, $x_i$ is the entry in
column/feature $i$ of $x$, and
 each $F_{i}$ is an arbitrary (easy to compute) function with range $R$.
\item
A collection ${\mathcal L} = \{ (G_1, L_1), \ldots (G_b, L_b) \}$ of additive inequalities,
where $G_i$ is a collection $\{g_{i,1}, g_{i,2}, \ldots g_{i, d} \}$ of $d$  (easy to compute) functions that map the column domains to the reals, 
 and each $L_i$ is a real number. 
 A row $x \in J$ satsifies the additive inequalities in
$\mathcal L$ if for all $i \in [1, b]$, it is the case that $\sum_{j=1}^d g_{i,j}(x_j) \le L_i$.
 
\end{itemize}
FAQ-AI($k$) is a special case of FAQ-AI when
the cardinality of $\mathcal L$ is at most $k$.

The output for the FAQ-AI problem is the result
of the query on the subset of the design matrix 
that satisfies the additive inequalities. 
That is, the output for the FAQ-AI instance with a SumSum query is:
\begin{align}
\label{equality:rfaqli-sumsum}
      Q({\mathcal L}(J)) = \bigoplus_{x \in {\mathcal L}(J)} \bigoplus_{i=1}^d F_{i}(x_i)
\end{align}
And the output for the FAQ-AI instance with a SumProd query is:
\begin{align}
\label{equality:rfaqli-sumprod}
    Q({\mathcal L}(J)) = \bigoplus_{x \in {\mathcal L}(J)} \bigotimes_{i=1}^d F_{i}(x_i)
\end{align}
Here ${\mathcal L}(J)$ is the set of  $x \in J$ that satisfy the additive inequalities in
$\mathcal L$. 
To aid the reader in appreciating these definitions, we now illustrate how some of the SVM related problems in the introduction can be reduced to FAQ-AI(1). 

\paragraph{Counting  the number of negatively labeled points correctly classified by a linear separator:}  Here each 
row $x$ of the design matrix $J$ conceptually consists of a point in $\mathbb{R}^{d-1}$, whose
coordinates are specified by the first $d-1$ columns in $J$, 
and a label in $\{1,-1\}$ in column $d$. Let  the  linear separator be defined by $\beta \in \mathbb{R}^{d-1}$.  A negatively labeled point $x$ is correctly classified if $\sum_{i=1}^{d-1} \beta_i x_i \leq 0$.  The number of such points can be counted using SumProd query with
one additive inequality as follows:
$\oplus$ is addition, 
$\otimes$ is multiplication, 
$F_i(x_i) = 1$ for all $i \in [d-1]$, 
$F_d(x_d) = 1$ if $x_d = -1$, and $F_d(x_d) = 0$ otherwise,
$g_{1,j}(x_j)= \beta_j x_j$ for $j \in [d-1]$, 
$g_{1,d}(x_d) = 0$, and
$L_1 =0$.

\paragraph{Finding the minimum distance to the linear separator of a correctly classified negatively labeled point:} 
This distance can be computed using a SumProd query with one additive inequality as follows:
$\oplus$ is the binary minimum operator, 
$\otimes$ is addition, 
$F_i(x_i) = \beta_i x_i$ for all $i \in [d-1]$, 
$F_d(x_d) = 1$ if $x_d = -1$, and $F_d(x_d) = 0$ otherwise,
$g_{1,j}(x_j)= \beta_j x_j$ for $j \in [d-1]$, 
$g_{1,d}(x_d) = 0$, and
$L_1 =0$.

\subsection{Related Results}

The Inside-Out algorithm~\cite{faq} can evaluate a SumProd query in time $O(m d^2 n^{h} \log n)$, where  $h$ is the fractional hypertree width~\cite{GM06} of the query (assuming unit time per semiring operation). 
Note that $h=1$ for the acyclic joins, and 
thus Inside-Out is a polynomial time algorithm for acyclic joins.
One can reduce SumSum queries to $m$ SumProd queries~\cite{faqai}, and thus they be solved in 
time $O(m^2 d^2 n^{h} \log n)$.
Conceptually the Inside-Out algorithm is a generalization of Dijkstra's shortest path algorithm. We explain Inside-Out in more detail in Section \ref{sect:insideout}. 
The Inside-Out algorithm builds on many earlier papers, including~\cite{Aji:2006,Dechter:1996,Kohlas:2008,GM06}.

FAQ-AI queries were introduced in \cite{faqai}. \cite{faqai} gave an algorithm with worst-case time
complexity $O(m d^2 n^{m/2} \log n)$. So this is
better than the standard practice of forming the design matrix, which has worst-case time complexity $\Omega(d n^m)$. Conceptually the improvement over standard practice arises because the algorithm  in \cite{faqai} avoids the last join.

Different flavors of queries with inequalities were also studied~\cite{Klug:1988:CQC:42267.42273,Koutris2017,DBLP:journals/corr/abs-1712-07445}.
 Relational algorithms for 
 linear/polynomial regression are considered in
 \cite{SystemF,khamis2018ac,IndatabaseLinearRegression,Kumar:2015:LGL:2723372.2723713,Kumar:2016:JJT:2882903.2882952}. An  algorithm for 
$k$-means on relational data is given in ~\cite{Rkmeans}.

\subsection{Our Results}

We start by showing in Section \ref{section:nphard1} that the FAQ-AI(1) problem is $\#P$-hard, even for the problem of counting the number of rows in the design matrix for a cross product join. Thus a relational
algorithm for FAQ-AI(1) queries
is extraordinarily unlikely as it would imply $P = \#P$.

Thus we turn to approximately computing FAQ-AI queries. 
An ideal result would be what we call 
a \emph{Relational Approximation
Scheme}  (RAS), which is a collection $\{ A_\epsilon \}$ of relational algorithms, one for each real $\epsilon > 0$, 
such that each $A_\epsilon$ is  outputs a solution  that has
relative error at most $\epsilon$. 
Our main result is a RAS for FAQ-AI(1) queries that have certain natural properties, which we now define. 

\begin{definition}~
\begin{itemize}
 \item
 An operator $\odot$ has bounded error if it is the case that when $x/(1+\delta_1) \le x' \le (1+\delta_1) x$ and
 $y/(1+\delta_2) \le y' \le (1+\delta_2) y$ then 
 $(x \odot y) / ((1+\delta_1)(1+\delta_2))\le x' \odot y' \le (1+\delta_1)(1+\delta_2) (x \odot y)$.
 \item
 An operator introduces no error if it is the case that  when $x/(1+\delta_1) \le x' \le (1+\delta_1) x$ and
 $y/(1+\delta_2) \le y' \le (1+\delta_2) y$ then 
 $(x \odot y) / (1+\max(\delta_1,\delta_2)) \le x' \odot y' \le (1+\max(\delta_1,\delta_2)) (x \odot y)$..
 \item
 An operator $\odot$  is repeatable if
for any two non-negative integers $k$ and $j$ and any non-negative real $\delta$ such that $k/(1+\delta) \le j \le (1+\delta) k$, it is the case that for every $x \in R$, $(\bigodot^{k} x) / (1+\delta) \le \bigodot^{j} x \le (1+\delta) \bigodot^{k} x$.
\item
An operator $\odot$ is
 monotone if it is either monotone increasing or monotone decreasing. The operator $\odot$ is
 monotone increasing if $x \odot y  \ge \max(x, y)$.
 The operator $\odot$ is monotone decreasing if $x \odot y \le \min(x, y)$.
\end{itemize}
\end{definition}

\begin{theorem}
\label{thm:sumsummain}
There is a RAS to compute a SumSum FAQ-AI(1) query over a 
commutative monoid 
$(R, \oplus, I_0)$ if:
\begin{itemize}
    \item The domain $R$ is a subset of reals $\mathbb{R}$.
    \item The operators $\oplus$ and $\otimes$ can be computed in polynomial time. 
    \item $\oplus$ introduces no error. 
    \item
    $\oplus$ is repeatable. 
\end{itemize}
\end{theorem}

\begin{theorem}
\label{thm:sumprodmain}
There is a RAS to compute a SumProd FAQ-AI(1) query  over 
 a commutative semiring
$(R, \oplus, \otimes, I_0, I_1)$ if: 
\begin{itemize}
\item 
    The domain $R$ is 
    $\mathbb{R}^+ \cup \{I_0\} \cup \{I_1\}$.
    \item
    $I_0, I_1 \in \mathbb{R}^+ \cup \{+\infty\} \cup \{ - \infty \}$
      \item The operators $\oplus$ and $\otimes$ can be computed in polynomial time. 
    \item $\oplus$ introduces no error.
    \item $\otimes$ has bounded error.
    \item
    $\oplus$ is monotone.  An operator $\odot$ is
 monotone if it is either monotone increasing or monotone decreasing. 
 \item
 The log of the aspect ratio of the query is polynomially bounded. The aspect ratio is the ratio of the maximum,  over every possible submatrix of the design matrix, of the
 value of the query on that submatrix, to the minimum,
 over every possible submatrix of the design matrix,
 of the
 value of the query on that submatrix. 
    \end{itemize}
\end{theorem}

In Section  \ref{sect:insideout}
we review the Inside-Out algorithm for SumProd queries over acyclic joins, as
our algorithms will use the Inside-Out algorithm.
    
In section \ref{section:rowcount},
we explain how to obtain a RAS for
a special type of FAQ-AI(1) query, an Inequality Row Counting Query, that
counts the number of rows
in the design matrix that satisfy a given
additive inequality.
A even more special case of an Inequality Row Counting query is counting the number of points that lie on a given side of a given hyperplane. 
Our algorithm for  Inequality Row Counting can be viewed as a reduction to the problem of 
evaluating a general SumProd query   (without any additive inequalities), over a more complicated type of semiring, that we call a \emph{dynamic programming semiring}. 
In a dynamic programming semiring the base elements can be thought of as arrays, and the summation and product operations are designed so that the SumProd query computes a desired dynamic program.
In the case of Inequality Row Counting,  
our SumProd query essentially ends up implementing the dynamic program for Knapsack Counting from \cite{Dyer03approximatecounting}.
Given the widespread utility of dynamic programming as a algorithm design technique, it seems to us likely that the use of dynamic programming semirings will be useful in designing relational algorithms for other problems. 
Connections between semirings and dynamic programming have certainly been observed in the past. But the references we could find mostly observed that some algorithms can be generalized to work over semirings, for example Dijkstra's algorithm is known to work over certain types of semirings~\cite{Mohri02}. We couldn't find references in the literature that designed semirings to compute particular dynamic programs (although it would hardly be shocking if such references were existed, and we welcome pointers if reviewers know of any such references).

The time complexity of our algorithm for Inequality Row Counting is at most $O\left(  \frac{m^6  \log^4 n}{\epsilon^2} \right)$
times the time complexity of the algorithm in \cite{faq} for evaluating an SumProd query without additive inequalities,
which is $O(d^2 m n \log n)$ if the join is acyclic. Thus for acyclic joins the running time of our algorithm is
$O\left(  \frac{m^7  d^2 n \log^5 n}{\epsilon^2} \right)$.
Note that in most natural instances of interest, $n$ is orders of magnitude larger than $d$ or $m$.
Thus it is important to note that the running time of our algorithm is nearly linear in the dominant parameter $n$.  Finally we end Section \ref{section:rowcount} by showing how use this Inequality Row Counting RAS to obtain a RAS for SumSum FAQ-AI(1) queries covered by Theorem \ref{thm:sumsummain}. 

In Section \ref{section:generalupper} we explain how to generalize our RAS  for SumProd FAQ-AI(1) queries covered by Theorem \ref{thm:sumprodmain}.

 In Section \ref{sect:applications} we show several applications of our main result to obtain RAS for several natural problems. And we give a few examples where our main result does not apply.

In section \ref{section:approximatehardness} we show that the problem of obtaining $O(1)$-approximation for 
a row counting query with two additive inequalities is NP-hard, even
for acyclic joins.  So this shows that our result for FAQ-AI(1) cannot be extended to FAQ-AI(2), and that a relational algorithm
with bounded relative error for row counting with two additive inequalities
 is quite unlikely, as such an algorithm would
imply $\mathrm{P}=\mathrm{NP}$.

\section{NP-Hardness of FAQ-AI(1)}
\label{section:nphard1}


\begin{theorem}
\label{theorem:nphardness:exact}
The problem of evaluating a FAQ-AI(1) query is $\#P$-Hard. 
\end{theorem}

\begin{proof}
We prove the $\#$-hardness by a reduction from the $\#P$-hard Knapsack Counting problem. An instance for Knapsack consists of a collection $W =\{w_1,\dots,w_d\}$ of nonnegative integer weights, and a nonnegative integer weight $C$. The output should be the number of subsets of $W$ with aggregate weight at most $C$.

 We construct the instance of FAQ-AI as follows. We creating $d$ tables. Each table $T_i$ has one column and  two rows,
 with entries $0$ and $w_i$. Then 
 $J = T_1 \Join T_2 \Join \dots \Join T_d$
is the cross product join of the tables.
We define $\beta$ to be the $d$ dimensional vector with $1$'s in all dimensions, and the additive inequality to $\beta \cdot x \leq C$. 
Then note that there is then a natural bijection between the rows in $J$ that satisfy this inequality the  subsets of $W$ with aggregate weight at most $C$. 
\end{proof}

\section{The Inside-Out Agorithm for Acyclic Queries}
\label{sect:insideout}

In this section, after some preliminary definitions, we explain how to obtain a hypertree decomposition of an acyclic join, and then explain (a variation of) the Inside-Out algorithm from~\cite{faq} for evaluating a SumProd query   $Q(J) =  \oplus_{x \in J} \otimes_{i=1}^d F_{i}(x_i)$ over a commutative semiring $(R,\oplus, \otimes, I_0, I_1)$ for acyclic join $J=T_1 \Join \dots \Join T_m$. 
A call of Inside-Out may optionally include a root table $T_r$. 

Let $C_i$ denote the set of features in $T_{i}$ and let $C = \bigcup_i C_i$. Furthermore, given a set of features $C_i$ and a tuple $x$, let $\Pi_{C_i}(x)$ be the projection of $x$ onto $C_i$.

\begin{definition}
\label{definition:acyclicity}
A join query $J=T_1 \Join \dots \Join T_m$ is acyclic if there exists a tree $G=(V,E)$, called the hypertree decomposition of $J$, such that:
\begin{itemize}
    \item 
The set of vertices are $V = \{v_1,\dots,v_m\}$, and \item
for every feature $c \in C$, the set of vertices $\{v_i| c \in C_i\}$ is a connected component of $G$.
\end{itemize}
\end{definition}

\paragraph{Algorithm to Compute Hypertree Decomposition:}
\begin{enumerate}
    \item Initialize graph $G=(V,\emptyset)$ where $V = \{v_1, \ldots, v_m \} $.
    \item Repeat the following steps until $|T| = 1$:
    \begin{enumerate}
        \item Find $T_i$ and $T_j$ in $T$ such that every feature of $T_i$ is either not in any other table of $T$ or is in $T_j$. If there exists no $T_i$ and $T_j$ with this property, then the query is cyclic.
        \item Remove $T_i$ from $T$ and add the edge $(v_i,v_j)$ to $G$.
    \end{enumerate}
\end{enumerate}

\paragraph{Inside-Out Algorithm:}
\begin{enumerate}
\item Compute the hypertree decomposition $G=(V, E)$ of $J$. 
    \item 
    Assign each feature $c$ in $J$ to an arbitrary table $T_i$ such that $c\in C_i$. Let $A_i$ denote the features assigned to $T_i$ in this step.
    \item For each table $T_i$, add a new column/feature $Q_i$. For all the tuples $x \in T_i$, initialize the value of column $Q_i$ in row $x$ to $Q_i^x =  \bigotimes_{j \in A_i} F_{j}(x_j)$. Note that if $A_i = \emptyset$ then $Q_i^x = I_1$.
    \item Repeat until $G$ has only one vertex \label{dijkstra}
    \begin{enumerate}
        \item Pick an arbitrary edge $(v_i, v_j)$ in  $G$ such that $v_i$ is a leaf and $i \neq r$.
        \item Let $C_{ij} = C_{i} \cap C_{j}$ be the shared features between $T_i$ and $T_j$.
        \item Construct a temporary table $T_{ij}$ that has the features $C_{ij} \cup \{Q_{ij}\}$. 
         \item \label{step:aggregation} If $C_{ij} = \emptyset$, then table $T_{ij}$ only has the column/feature $Q_{ij}$ and one row, and its lone entry is set to $\oplus_{x\in T_i} Q_i^x$. Otherwise, iterate through the $y$ such that there exists an $x \in T_i$ for which $\Pi_{C_i}(x)=y$, and add the row $(y,Q_{ij}^y)$ to table $T_{ij}$ where:
         \begin{enumerate}
             \item $Q_{ij}^y$ is set to the sum, over all rows $x \in T_i$ such that $C_{ij}(x)=y$, of   $Q_i^x$. \label{binarysumstep}
\end{enumerate}
        \item \label{step:multiplication} For all the tuples $(x,Q_j^x) \in T_j$, let $y = \Pi_{C_{ij}}(x)$, and update $Q_j^x$ by
        \begin{align*}
            Q_j^x \gets Q_j^x \otimes Q_{ij}^y.
        \end{align*}
        If $(y,Q_{ij}^y) \notin T_{ij}$, set $Q_j^x = I_0$. 
        \item Remove vertex $v_i$ and edge $(v_i, v_j)$ from $G$.
    \end{enumerate}
    \item \label{step:roottable} At the end, when there is one vertex $v_r$ left in $G$, return the value 
    \begin{align*}
        \bigoplus_{(x,Q_r^x) \in T_r} Q_r^x
    \end{align*}
\end{enumerate}

When we use the Inside-Out algorithm in context of an approximation algorithm is important that the sum computed in step \ref{binarysumstep} is computed using a balanced binary tree, so that if $k$ items are being summed, the depth of the expression tree is at most $\lceil \log k \rceil$.

One way to think about step \ref{dijkstra} of the algorithm is that it is  updating the $Q_j$ values to what they would
be if, what good old CLRS~\cite{Cormen2001} calls a relaxation in the description of the Bellman-Ford shortest path algorithm, was applied to every edge in a particular bipartite graph $G_{i,j}$.
 In $G_{i,j}$ one side of the vertices are the rows in $T_i$,  and the other side are the rows in $T_j$,
and there a directed edge $(x, y)$  from a vertex/row in $T_i$ to a vertex/row in $T_j$ if they have
equal projections onto $C_{i,j}$.
The length $P_j^y$ of edge $(x, y)$ is the original
value of $Q_j^y$ before the execution of step \ref{dijkstra}.
A relaxation step on a directed edge $(x, y)$
is then $Q_j^y = (P_j^y \otimes Q_i^x) \oplus Q_j^y$. So the result of step \ref{dijkstra} of
Inside-Out is the same as relaxing every edge in $G_{i,j}$. Although  Inside-Out doesn't explicitly relax every edge. Inside-Out exploits the 
structure of $G_{i,j}$, by grouping together rows
in $T_i$ that have the same projection onto $C_{i,j}$, to be more efficient.

\section{Algorithm for Inequality Row Counting}
\label{section:rowcount}

The {\em Inequality Row Counting} is a special case of SumProd FAQ-AI(1) in which the SumProd query $Q({\mathcal L}(J)) = \sum_{x \in J} \prod_{i = 1}^d 1$ counts the number of rows in the design matrix that satisfy the constraints $\mathcal{L}$, which consists of one additive constraint $\sum_i g_i(x_i) \leq L$,
over a join $J = T_1 \Join T_2 \Join \dots \Join T_m$. In subsection \ref{section:rowcount:exact}, we
design a SumProd query over a dynamic programming semiring that computes 
$Q({\mathcal L}(J))$ exactly in exponential time. 
 Then in subsection \ref{subsect:rowcountsketch} we explain how to apply standard sketching techniques to obtain a RAS. 
We finish in subsection \ref{subsect:generalupper:sumsum} by explaining how to use our algorithm for Inequality Row Counting to obtain to obtain a RAS SumSum queries covered by Theorem \ref{thm:sumsummain}. 

\subsection{An Exact Algorithm}
\label{section:rowcount:exact}

We first define a commutative semiring $(S, \sqcup, \sqcap, E_0, E_1)$ as follows:
\begin{itemize}
    \item 
The elements of the base  set $S$ are finite multi-sets of real numbers. Let $\# A (e)$ denote the frequency of the real value $e$ in the multi-set $A$ and let it be $0$ if $e$ is not in $A$. Thus one can also think of $A$ as a set of pairs of the form $(e, \# A (e))$. 

\item
The additive identity $E_0$ is the empty set $ \emptyset$.
\item
The addition operator
$\sqcup$ is the union of the two multi-sets; that is $A=B \sqcup C$ if and only if for all real values $e$, $\# A (e) = \# B(e) + \#C(e)$. 
\item
The multiplicative identity is $E_1 = \{0\}$.
\item
The multiplication operator $\sqcap$ contains the  pairwise sums from the two input multi-sets; that is, $A = B \sqcap C$ if and only if for all real values $e$, $\# A (e) = \sum_{i \in \mathbb{R}}(\# B(e-i) \cdot \# C(i))$. Note that this summation is well-defined because there is a finite number of values for $i$ such that $B(e-i)$ and $C(i)$ are non-zero.
\end{itemize}

\begin{lemma}
\label{lemma:rowcounting:semiring}
$(S, \sqcup, \sqcap, E_0, E_1)$ is a commutative semiring. 
\end{lemma}

\begin{proof}
    To prove the lemma, we prove the following claims in the order in which they appear.
    \begin{enumerate}
        \item $A \sqcup B = B \sqcup A$
        \item $A \sqcup (B \sqcup C) =( A \sqcup B) \sqcup C$
        \item $A \sqcup E_0 = A$
        \item $A \sqcap B = B \sqcap A$
        \item $A \sqcap (B \sqcap C) =( A \sqcap B) \sqcap C$
        \item $A \sqcap E_0 = E_0$
        \item $A \sqcap E_1 = A$
        \item $A \sqcap (B \sqcup C) = (A \sqcap B) \sqcup (A \sqcap C)$
    \end{enumerate}
    
    First we show $\sqcup$ is commutative and associative and $A \sqcup E_0 = A$.
    By definition of $\sqcup$, $C=A \sqcup B$ if and only if for all $e \in C$ we have $\# C(e) = \# A(e) + \# B(e)$. Since summation is commutative and associative, $\sqcup$ would be commutative and associative as well. Also note that if $B = E_0 = \emptyset$ then $\# B(e) = 0$ for all values of $e$ and as a result $\# C(e) = \# A(e)$ which means $C=A$.
    
    Now we can show that $\sqcap$ is commutative and associative. By definition of $\sqcap$, $C = A \sqcap B$ if and only if for all values of $e$, $\# C(e) = \sum_{i\in \mathbb{R}}(\# A(e-i) \cdot \# B(i))$, since we are taking the summation over all values:
    \begin{align*}
        \# C(e) &= \sum_{i \in \mathbb{R}}(\# A(e-i) \cdot \# B(i))
        \\ &= \sum_{i \in \mathbb{R}}(\# A(i) \cdot \# B(e-i))
    \end{align*}
    The last line is due to the definition of $B \sqcap A$, which means $\sqcap$ is commutative. 
    
    To show claim (5), let $D = A\sqcap(B\sqcap C)$ and $D' = (A \sqcap B) \sqcap C$:
    \begin{align*}
        \# D(e) &= \sum_{i \in \mathbb{R}} \# A(e-i) \cdot \Big (\sum_{j \in \mathbb{R}} \# B(i-j) \cdot \# C(j)\Big )
        \\
        &= \sum_{i,j \in \mathbb{R}} \# A(e-i) \cdot \# B(i-j) \cdot \# C(j)
    \end{align*}
    By setting $i' = e - j$ and $j' = e- i$, we obtain:
    \begin{align*}
        \# D(e)
        &= \sum_{i',j' \in \mathbb{R}} \# A(j') \cdot \# B(i'-j') \cdot \# C(e-i')
        \\
        &= \sum_{i' \in \mathbb{R}} \Big(\sum_{j' \in \mathbb{R}}\# A(j') \cdot \# B(i'-j') \Big) \cdot \# C(e-i') = \# D'(e),
    \end{align*}
    which means $\sqcap$ is associative, as desired.
    
    Now we prove $A\sqcap E_0 = E_0$ and $A \sqcap E_1 = A$. The claim (6) is easy to show since for all $e$, $\# E_0(e) = 0$ then for all real values $e$ $\sum_{i\in \text{dist}(A)}(\# A(i) \cdot \# E_0(e-i)) = 0$; therefore, $A\sqcap E_0 = E_0$. For claim (7), we have $\sum_{i\in \text{dist}(E_1)}(\# E_1(i) \cdot \# A(e-i)) = (\# E_1(0) \cdot \# A(e)) = \# A(e)$; therefore, $A\sqcap E_1 = A$.
    
    At the end all we need to show is the distributive law which means we need to show $A\sqcap(B \sqcup C) = (A\sqcap B) \sqcup (A \sqcap C)$. Let $D =A\sqcap(B \sqcup C)$ and $D'=(A\sqcap B) \sqcup (A \sqcap C)$. We have,
    \begin{align*}
        \# D(e) &= \sum_{i \in \mathbb{R}} \# A(e-i) \cdot (\# B(i) + \# C(i))
        \\ &= \sum_{i \in \mathbb{R}} (\# A(e-i) \cdot \# B(i)) + (\# A(e-i) \cdot \# C(i))
        \\ &= \sum_{i \in \mathbb{R}} (\# A(e-i) \cdot \# B(i)) + \sum_{j \in \mathbb{R}}(\# A(e-j) \cdot \# C(j)) = \# D'(e)
    \end{align*}
\end{proof}

\begin{lemma}
\label{lemma:multiseturn}
    The SumProd query $\widehat Q(J) = \sqcup_{x \in J} \sqcap_{i = 1}^d F_i(x_i)$, where   $F_i(x_i) = \{ g_i(x_i) \}$, evaluates to the  multiset  $\{\sum_i g_i(x_i) \mid x \in J\}$ the aggregate of the $g_i$ functions over the rows of $J$. 
\end{lemma}

\begin{proof}
    Based on the definition of $\sqcap$ we have,
    \begin{align*}
        \sqcap_{i=1}^d F_i(x_i) &= \sqcap_{i=1}^d \{g_i(x_i)\}
       = \left\{ \sum_{i=1}^d g_i(x_i)\right\}
    \end{align*}
    Then we can conclude:
    \begin{align*}
        \sqcup_{x \in J} \sqcap_{i = 1}^d F_i(x_i) &= \sqcup_{x \in J} \left\{\sum_{i=1}^d g_i(x_i) \right\} 
       = \left\{ \sum_{i=1}^d g_i(x_i) \mid x \in J\right\}
    \end{align*}
\end{proof}

Thus the inequality row count is the number of elements in the multiset returned by $\widehat Q(J)$ that are at most $L$.

\subsection{Applying Sketching}
\label{subsect:rowcountsketch}

    For a multiset $A$, let  $\triangle A(t)$ denote the number of elements in $A$ that are less than or equal to $t$. 
Then the $\epsilon$-sketch $\mathbb{S}_\epsilon(A)$ of a multiset $A$
is a multiset formed in the following manner: For each integer $k \in [1, \lfloor \log_{1+\epsilon} |A| \rfloor ]$ there are $\lfloor (1+\epsilon)^k \rfloor - \lfloor (1+\epsilon)^{k-1} \rfloor$ copies of the
$\lfloor (1+\epsilon)^k \rfloor$ smallest element $x_k \in A$; that is, 
$x_k = \triangle A(\lfloor (1+\epsilon)^k \rfloor)$.
Note that $|\mathbb{S}_\epsilon(A)|$ may be less
that $|A|$ as the maximum value of $k$ is $\lfloor \log_{1+\epsilon} |A| \rfloor $.
We will show in Lemma 
\ref{lemma:approx:guarantee:single:sketch}
that sketching preserves  $\triangle A(t)$ 
within $(1 + \epsilon)$ factor.

\begin{lemma}
\label{lemma:approx:guarantee:single:sketch}
    For all $t \in \mathbb{R}$, we have $(1-\epsilon) \triangle A(t) \leq \triangle \mathbb{S}_\epsilon(A)(t) \leq \triangle A(t)$.
\end{lemma}

\begin{proof}
    Let $A' = \mathbb{S}_{\epsilon}(A)$. Note that since we are always rounding the weights up, every item in $A$ that is larger than $t$ will be larger in $A'$ as well. Therefore, $\triangle A'(t) \leq \triangle A(t)$. We now show the lower bound. Recall that in the sketch, every item in the sorted array $A$ with an index in the interval $((1+\epsilon)^{i}, (1+\epsilon)^{i+1}]$ (or equivalently $((1+\epsilon)^{i}, \lfloor (1+\epsilon)^{i+1} \rfloor$) will be rounded to $A[\lfloor(1+\epsilon)^{i+1} \rfloor]$. Let $i$ be the integer such that $(1+\epsilon)^{i} < \triangle A(t) \leq (1+\epsilon)^{i+1}$, then the only items that are smaller or equal to $t$ in $A$ and are rounded to have a weight greater than $t$ in $A'$ are the ones with index between $(1+\epsilon)^{i}$ and $j = \triangle A(t)$. Therefore,
    \begin{align*}
        \triangle A(t) - \triangle A'(t) \leq& j-(1+\epsilon)^{i} \leq (1+\epsilon)^{i+1} - (1+\epsilon)^{i} 
        \\=& \epsilon (1+\epsilon)^{i} \leq \epsilon \triangle A(t),
    \end{align*}
    which shows the lower bound of $\# \triangle A(t)$ as claimed.
\end{proof}

Then our algorithm runs the Inside-Out algorithm, with the operation $\sqcup$ replaced by an operation $\bigcirc$, defined by
$A \bigcirc B = \mathbb{S}_\alpha( A \sqcup B)$, and with the operation $\sqcap$ replaced by an operation $\odot$, defined by $A \odot B = \mathbb{S}_\alpha(A \sqcap B)$, where $\alpha = \Theta(\frac{\epsilon}{m^2 \log(n)})$. That is, the operations $\bigcirc$ and $\odot$ are the sketched versions of $\sqcup$ and $\sqcap$.   That is, Inside-Out is run on the query  $\tilde Q(J) = \bigcirc_{x \in J} \odot_{i = 1}^d F_i(x_i)$, where   $F_i(x_i) = \{ g_i(x_i) \}$.

Because 
 $\bigcirc$ and $\odot$ do not necessarily form a semiring, Inside-Outside may not  return
$\tilde Q(J)$.
However, Lemma \ref{lemma:epsilonsketch} bounds the error introduced by each application of $\bigcirc$ and $\odot$. 
This makes it possible in
Theorem~\ref{thm:counting} to bound the error of Inside-Out's output.

\begin{lemma}
\label{lemma:epsilonsketch}
    Let $A'= \mathbb{S}_{\beta}(A)$,  $B'= \mathbb{S}_\gamma(B)$,  $C=A \sqcup B$, $C'=A' \sqcup B'$, $D=A \sqcap B$, and $D'=A' \sqcap B'$. For all $t \in \mathbb{R}$, we have:
    \begin{enumerate}
        \item $(1-\max(\beta,\gamma)) \triangle C(t) \leq \triangle C'(t) \leq \triangle C(t)$
        \item $(1-\beta - \gamma) \triangle D(t) \leq \triangle D'(t) \leq \triangle D(t)$
    \end{enumerate}
\end{lemma}

\begin{proof}
    By the definition of $\sqcup$, we know $\# C'(t) = \# A'(t) + \# B'(t)$; thus we have:
    \begin{align*}
        \triangle C'(t) &=\sum_{\tau \leq t} \# C'(\tau) = \sum_{\tau\leq t} \# A'(\tau) + \sum_{\tau\leq t} \# B'(\tau)
        \\
        &= \triangle A'(t) + \triangle B'(t)
    \end{align*}
    Similarly, we have $\triangle C(t) = \triangle A(t) + \triangle B(t)$. Then by Lemma \ref{lemma:approx:guarantee:single:sketch} we immediately have the first claim.
    
    Let $D'' = A \sqcap B'$, Based on the definition of $\sqcap$ we have:
    \begin{align*}
        \triangle D''(t) &= \sum_{\tau \leq t} \# D''(t) = \sum_{\tau \leq t} \sum_{v \in \mathbb{R}}(\# A(v) \cdot \# B'(\tau-v))
        \\
        &=  \sum_{v \in \mathbb{R}}\sum_{\tau \leq t}(\# A(v) \cdot \# B'(\tau-v))
        =\sum_{v \in \mathbb{R}}(\# A(v) \cdot \triangle B'(t-v))
    \end{align*}
    Therefore, using Lemma \ref{lemma:approx:guarantee:single:sketch} we have:
    $(1-\gamma) \triangle D(t) \leq \triangle D''(t) \leq (1+\gamma) \triangle D(t)$.
    We can similarly replace $A$ to $A'$ in $D''$ and get $(1-\beta) \triangle D''(t) \leq \triangle D'(t) \leq (1+\beta) \triangle D''(t)$
which proves the second claim. 
\end{proof}

\begin{theorem}
    \label{thm:counting}
Our algorithm achieves an $(1+\epsilon)$-approximation to the Row Count Inequality query $Q({\mathcal L}(J))$
in time
$O(\frac{1}{\epsilon^2}(m^3 \log^2(n))^2  (d^2 m n^{h} \log(n)) )$. 
\end{theorem}

\begin{proof}
We first consider the approximation ratio. 
Inside-Out on the query $\tilde Q(J)$ performs
the same semiring operations as does on the 
query 
$\widehat{Q}(J)$, but it additionally applies
the  $\alpha$-Sketching operation over each partial results, meaning the algorithm applies $\alpha$-Sketching after steps \ref{step:aggregation},\ref{step:multiplication}, and \ref{step:roottable}. Lets look at each iteration of applying steps \ref{step:aggregation} and \ref{step:multiplication}. Each value produced in the steps \ref{step:aggregation} and \ref{step:roottable} is the result of applying $\bigcirc$ over at most $n^m$ different values (for acyclic queries it is at most $n$). Using Lemma \ref{lemma:epsilonsketch} and the fact that the algorithm applies $\bigcirc$ first on each pair and then recursively on each pair of the results, the total accumulated error produced by each execution of steps \ref{step:aggregation} and \ref{step:roottable} is $m\log(n) \alpha$. Then, since the steps \ref{step:aggregation} and \ref{step:multiplication} will be applied once for each table, and \ref{step:multiplication} accumulates the errors produced for all the tables, the result of the query will be $(m^2\log(n)+m)\alpha$-Sketch of $\widehat{Q}(J)$.

We now turn to bounding the running time of our algorithm.
The time complexity of Inside-Out is $O(m d^2 n^h \log n)$ when the summation and product operators take a constant time~\cite{faq}. Since each multi-set in a partial result  has at most $m^n$ members in it, an $\alpha$-sketch of the partial results will have at most $O(\frac{m\log n}{\alpha})$ values, and we compute each of $\bigcirc$ and
$\odot$ in time $O\left(\left(\frac{m\log n}{\alpha}\right)^2\right)$. 
Therefore our algorithm runs in time $O(\frac{1}{\epsilon^2}(m^3 \log^2(n))^2  (d^2 m n^{h} \log(n)) )$.
\end{proof}

\subsection{SumSum FAQ-AI(1)}
\label{subsect:generalupper:sumsum}

In this subsection we prove Theorem  \ref{thm:sumsummain}, that there is a RAS for  SumSum FAQ-AI(1) queries covered by the theorem. Our algorithm for SumSum queries uses our algorithm  for Inequality Row Counting.  Consider the SumSum $Q({\mathcal L}(J)) = \oplus_{x \in {\mathcal L}(J)} \oplus_{i=1}^d F_{i}(x_i)$, where 
$\mathcal{L}$ consists of one additive constraint $\sum_i g_i(x_i) \leq L$,
over a join $J = T_1 \Join T_2 \Join \dots \Join T_m$.

\paragraph{SumSum Algorithm:} 
For each table $T_j$, we run our Inside-Out on the Inequality Row Counting query $ \tilde Q(J)$, 
with the root table being $T_j$, and let $\tilde T_j$
be the resulting table just before step \ref{step:roottable} of Inside-Out is executed.  From the resulting tables, one can compute,  for every feature $i \in [d]$ and for each possible  value of $x_i$ for $x \in J$,  a $(1+\epsilon)$-approximation  $U(x_i)$ to the number of rows in the design matrix that contain value $x_i$ in column $i$, by aggregating over the row counts in any table $\tilde T_j$ that contains feature $i$. Then one can evaluate $Q({\mathcal L}(J))$ by
\begin{align*}
\bigoplus_{i=1}^d \bigoplus_{x_i \in D(i)} \bigoplus_{j=1}^{U(x_i)} F_{i}(x_i) 
\end{align*}
where $D(i)$ is the domain of feature $i$. 

Note that $\oplus$ operator is assumed to be repeatable, meaning if we have a $(1 \pm \epsilon)$ approximation of $U(x_i)$ then the approximation error of $\bigoplus_{j=1}^{U(x_i)} F_{i}(x_i)$ is also $(1 \pm \epsilon)$. Therefore, our SumSum algorithm is a $(1+\epsilon)$-approximation algorithm because the only error introduced is by our Inequality Row Counting algorithm. The running time is
$O(\frac{1}{\epsilon^2}(m^3 \log^2(n))^2  (d^2 m^2 n^{h} \log(n)) )$ because we run $m$  Inequality Row Counting algorithm $m$ times.

\section{SumProd FAQ-AI(1)}
\label{section:generalupper}

In this section we prove Theorem  \ref{thm:sumprodmain}, that there is a RAS for  SumProd FAQ-AI(1) queries covered by the theorem. Our RAS for such queries 
generalizes our RAS for Inequality Row Counting. 
Consider the SumProd query
$    Q({\mathcal L}(J)) = \oplus_{x \in {\mathcal L}(J)} \otimes_{i=1}^d F_{i}(x_i).
$
where $\mathcal L$ consists of the single additive  constraint $\sum_{i=1}^d g_i(x_i) \le L$.
We again first give an exact algorithm that can viewed as a reduction to a SumProd query over a dynamic programming semiring, and then apply sketching. 

\subsection{An Exact Algorithm}
\label{section:sumprod:exact}

We first define a structured commutative semiring $(S, \sqcup, \sqcap, E_0, E_1)$
derived from the $(R, \oplus, \otimes, I_0, I_1)$ as follows:
\begin{itemize}
\item
The base set $S$ are finite subsets $A$ of $\mathbb{R} \times (R - \{I_0\})$ 
with the property that $(e, v) \in A$ and $(e, u) \in A$ implies $v=u$; so there is only one tuple in
$A$ of the form $(e, *)$. 
One can can interpret the value of $v$ in a tuple $(e, v) \in A$ as a (perhaps fractional) multiplicity of $e$.
\item
The additive identity  $E_0$ is the empty set $\emptyset$.
\item
The multiplicative identity  $E_1$ is $\{(0,I_1)\}$.
\item
For all  $e\in \mathbb{R}$, define $\# A(e)$ to be $ v$ if  $(e, v) \in A$ and $I_0$ otherwise.
\item
The addition operator $\sqcup$ is defined by $A \sqcup B = C$ if and only if for all  $e \in \mathbb{R}$, it is the case that  $\# C(e) = \# A(e) \oplus \# B(e)$.   
 \item 
 The multiplication operator $\sqcap$ is defined by 
 $A \sqcap B = C$ if and only if for all  $e \in \mathbb{R}$, it is the case that $\# C(e) = \bigoplus_{i \in \mathbb{R}} \# A(e-i) \otimes \# B(i)$.
\end{itemize}

\begin{lemma}
\label{lemma:generated:semiring}
If $(R, \oplus, \otimes, I_0, I_1)$, is a commutative semiring then  $(S, \sqcup, \sqcap, E_0, E_1)$ is a commutative semiring.
\end{lemma}

\begin{proof}
we prove the following claims respectively:
    \begin{enumerate}
        \item $A \sqcup B = B \sqcup A$
        \item $A \sqcup (B \sqcup C) =( A \sqcup B) \sqcup C$
        \item $A \sqcup E_0 = A$
        \item $A \sqcap B = B \sqcap A$
        \item $A \sqcap (B \sqcap C) =( A \sqcap B) \sqcap C$
        \item $A \sqcap E_0 = E_0$
        \item $A \sqcap E_1 = A$
        \item $A \sqcap (B \sqcup C) = (A \sqcap B) \sqcup (A \sqcap C)$
    \end{enumerate}
    
    Based on the definition, $C = A \sqcup B$ if and only if $\# C(e) = \# A(e) \oplus \# B(e)$; since $\oplus$ is commutative and associative, $\sqcup$ will be commutative and associative as well. Furthermore, $\# A(e) \oplus \# E_0(e) = \# A(e) \oplus I_0 = \# A(e)$; therefore, $A \sqcup E_0 = A$.
    
    Let $C= A \sqcap B$, using the commutative property of $\otimes$ and change of variables, we can prove the fourth claim:
    \begin{align*}
        \# C(e) &= \bigoplus_{i \in \mathbb{R}} \# A(e-i) \otimes \# B(i)
        \\ &= \bigoplus_{i \in \mathbb{R}} \# B(i) \otimes \# A(e-i)
        \\ &= \bigoplus_{j \in \mathbb{R}} \# B(e-j) \otimes \# A(j)
    \end{align*}
    
    Similarly, using change of variables $j'=i-j$ and $i'=j$, and semiring properties of the $\otimes$ and $\oplus$, we have:
    \begin{align*}
        & \bigoplus_{i \in \mathbb{R}} \# A(e-i) \otimes (\bigoplus_{j \in \mathbb{R}} \# B(i-j) \otimes \# C(j))
        \\
        =&\bigoplus_{i \in \mathbb{R}} \bigoplus_{j \in \mathbb{R}} (\# A(e-i) \otimes  \# B(i-j)) \otimes \# C(j)
        \\
        =&\bigoplus_{i' \in \mathbb{R}} (\bigoplus_{j' \in \mathbb{R}} \# A(e-i'-j') \otimes  \# B(j')) \otimes \# C(i')
    \end{align*}
    
    Therefore, $A\sqcap (B\sqcap C) = (A \sqcap B) \sqcap C$.
    
    The claim $A \sqcap E_0 = E_0$ can be proved by the fact that $\# E_0(i) = I_0$ for all the elements and $\# A(e-i) \otimes I_0 = I_0$. Also we have $\bigoplus_{i \in \mathbb{R}} \# A(e-i) \otimes \# E_1(i) = \# A(e)$ because, $\# E_1(i) = I_0$ for all nonzero values of $i$ and it is $I_1$ for $e=0$; therefore, $A \sqcap E_1 = A$.
    
    Let $D = A \sqcap (B \sqcup C)$, the last claim can be proved by the following:
    \begin{align*}
        \# D(e) &= \bigoplus_{i \in \mathbb{R}} \# A(e-i) \otimes (\# B(i) \oplus \# C(i))
        \\
        &= \bigoplus_{i \in \mathbb{R}} \big( (\# A(e-i) \otimes \#B(i)) \oplus (\# A(e-i) \otimes \# C(i)) \big)
        \\
        &= \big(  \bigoplus_{i \in \mathbb{R}} (\# A(e-i) \otimes \#B(i)) \big) \oplus \big( \bigoplus_{i \in \mathbb{R}} (\# A(e-i) \otimes \#C(i)) \big) 
    \end{align*}
    where the last line is the definition of $(A \sqcap B) \sqcup (A \sqcap C)$.
\end{proof}

 For each column $i  \in [d]$, we define the function $\mathcal{F}_{i}$ to be
$ \{(g_{i}(x_i) ,  F_{i}(x_i) ) \}$ if $F_{i}(x_i) \neq I_0$  and the empty set otherwise. 
 Our algorithm for computing $Q({\mathcal L}(J))$ runs the Inside-Out algorithm on the SumProd query:
\begin{align*}
    \widehat{Q} = \sqcup_{x \in J} \sqcap_{i=1}^d \mathcal{F}_{i}(x_i)
\end{align*}
and returns $\bigoplus_{e \leq L} \# \widehat{Q}(e)$.

\begin{lemma}
\label{lemma:resultofalgorithm}
This algorithm correctly computes $Q({\mathcal L}(J))$.
\end{lemma}

\begin{proof}
    We can rewrite the generated FAQ as follow:
    \begin{align*}
        \hat{Q} &= \sqcup_{x \in J} \sqcap_{i=1}^d \mathcal{F}_{i}(x_i)
        \\ &= \sqcup_{x \in J} \sqcap_{i=1}^d \{ ( g_{i}(x_i) , F_{i}(x_i)  ) \}
        \\
        &= \sqcup_{x \in J}  \left\{ \left( \sum_{i=1}^d g_{i}(x_i) , \bigotimes_{i=1}^d F_{i}(x_i)  \right) \right\}
    \end{align*}
    Then the operator $\sqcup$ returns a set of pairs $(e,v)$ such that for each value $e$, $v = \# \hat{Q}(e)$ is the aggregation using $\oplus$ operator over the rows of $J$ where $\sum_{i=1}^d g_{i}(x_i) = e$. More formally,
    \begin{align*}
        \# \hat{Q}(e) = \bigoplus_{x \in J, \sum_{i} g_{i}(x_i) = e} \left(\bigotimes_{i=1}^d F_{i}(x_i)\right)
    \end{align*}
    Therefore, the value returned by the algorithm is
    \begin{align*}
        \bigoplus_{e \leq L} \# \hat{Q}(e)
        =& \bigoplus_{e \leq L} \bigoplus_{x \in J, \sum_{i} g_{i}(x_i) = e} \left(\bigotimes_{i=1}^d F_{i}(x_i)\right)
        \\ =& \bigoplus_{x \in \mathcal{L}(J)} \bigotimes_{i=1}^d F_{i}(x_i)
    \end{align*}
\end{proof}

\subsection{Applying Sketching}
\label{subsect:sumprodsketch}

For a set $A\in S$ define $\triangle A(\ell)$ to be $\oplus_{e \leq \ell} \# A(e)$. Note that $\triangle A(\ell)$ will be monotonically increasing if $\oplus$ is monotonically increasing, and it will be monotonically decreasing if $\oplus$ is monotonically decreasing.

Conceptually an $\epsilon$-sketch $\mathbb{S}_\epsilon(A)$ of an element $A \in S$ rounds all multiplicities up to an integer power of $(1+\epsilon)$.
Formally the $\epsilon$-sketch $	\mathbb{S}_\epsilon(A)$ of $A$ is the element $A'$ of $S$ satisfying
\begin{align*}
    \# A'(e) = \begin{cases}
    \bigoplus_{L_k < e \leq U_k} \# A(e) & \text{if } \exists k \; e=U_k
    \\
    I_0 & \text{otherwise}
    \end{cases}
\end{align*}
where
$$L_0 = \min \{e \in \mathbb{R} \vert \: \triangle A(e) \leq 0 \}$$ 
and for $k \ne 0$
$$L_k = \min \{e \in \mathbb{R} \vert \: \rho (1+\epsilon)^{k-1} \leq \triangle A(e) \leq \rho (1+\epsilon)^{k} \}$$
and where
$$U_0 = \max \{e \in \mathbb{R} \vert \: \triangle A(e) \leq 0 \}$$
and for $k \ne 0$
$$U_k = \max \{e \in \mathbb{R} \vert \: \rho (1+\epsilon)^{k-1} \leq \triangle A(e) \leq \rho (1+\epsilon)^{k} \}$$
where
$\rho = \min \{\# A(e) \:\vert\: \#  e \in \mathbb{R}\text{ and } \# A(e)>0\}$. For the special case that $\# A(e) \leq 0$ for all $e \in \mathbb{R}$, we only have $L_0$ and $U_0$. Note that the only elements of $R$ that can be zero or negative are $I_0$ and $I_1$; therefore, in this special case, $\# A(e)$ for all the elements $e$ is either $I_0$ or $I_1$.

\begin{lemma}
\label{Lemma:General:Sketching}
For all $A \in S$, for all $e \in \mathbb{R}^+$, if $A' = \mathbb{S}_\epsilon(A)$ then $$\triangle A(e)/(1+\epsilon) \leq \triangle A'(e) \leq   (1+\epsilon)\triangle A(e)$$
\end{lemma}

\begin{proof}
Since $\triangle A(e)$ is monotone, the intervals $[L_k, U_k]$ do not have any overlap except over the points $L_k$ and $U_k$, and if the $\triangle A(e)$ is monotonically increasing, then $L_k = U_{k+1}$; and if $\triangle A(e)$ is monotonically decreasing, then $L_k = U_{k-1}$. 

For any integer $j$ we have:
\begin{align}
\label{equality:sketch:general}
\begin{split}
    \triangle A(U_j)
    &=\bigoplus_{i \leq U_j}\# A(i) 
    = \bigoplus_{k\leq j} \bigoplus_{L_k <i \leq U_k}\# A(i)
    = \bigoplus_{k\leq j} \# A'(U_k)
    \\
    &= \bigoplus_{i \leq U_j} \# A'(i)
    =  \triangle A'(U_j)
\end{split}
\end{align}

Now, first we assume $\triangle A(e)$ is monotonically increasing and prove the lemma. After that, we do the same for the monotonically decreasing case.
Given a real value $e$, let $k$ be the integer such that $e \in (L_k, U_k]$. Then using the definition of $U_k$ and Equality \eqref{equality:sketch:general} we have:
\begin{align*}
    \triangle A(e)/(1+\epsilon) &\leq
    \triangle A(U_k)/(1+\epsilon) = \triangle A(U_{k-1})
    = \triangle A'(U_{k-1})
    \\
    &\leq \triangle A'(e) 
    \leq \triangle A'(U_{k}) 
    = \triangle A(U_{k})
    \\
    &= (1+\epsilon) \triangle A(U_{k-1})
    \leq (1+\epsilon) \triangle A(e)
\end{align*}
Note that in the above inequalities, for the special case of $k=0$, we can use $L_k$ instead of $U_{k-1}$.
Similarly for monotonically decreasing case we have:
\begin{align*}
    \triangle A(e)/(1+\epsilon) &\leq
    \triangle A(U_{k-1})/(1+\epsilon) = \triangle A(U_{k})
    = \triangle A'(U_{k})
    \\
    &\leq \triangle A'(e) 
    \leq \triangle A'(U_{k-1}) 
    = \triangle A(U_{k-1})
    \\
    &= (1+\epsilon) \triangle A(U_{k})
    \leq (1+\epsilon) \triangle A(e)
\end{align*}
\end{proof}

Then our algorithm runs the Inside-Out algorithm, with the operation $\sqcup$ replaced by an operation $\bigcirc$, defined by
$A \bigcirc B = \mathbb{S}_\alpha( A \sqcup B)$, and with the operation $\sqcap$ replaced by an operation $\odot$, defined by $A \odot B = \mathbb{S}_\alpha(A \sqcap B)$, where $\alpha = \frac{\epsilon}{m^2 \log(n)+m}$. That is, the operations $\bigcirc$ and $\odot$ are the sketched versions of $\sqcup$ and $\sqcap$. 
Our algorithm returns $\triangle A(L) = \oplus_{e\leq L} \# A(e)$.

Because 
 $\bigcirc$ and $\odot$ do not necessarily form a semiring, Inside-Outside may not  return
$\bigcirc_{x \in J} \odot_{i = 1}^m F_i(x_i)$.
However, Lemma \ref{Lemma:General:Operator:Approx} bounds the error introduced by each application of $\bigcirc$ and $\odot$. 
This makes it possible in
Theorem~\ref{thm:sumprodmain} to bound the error of Inside-Out's output.

\begin{lemma}
\label{Lemma:General:Operator:Approx}
Let $A'= \mathbb{S}_{\beta}(A)$,  $B'= \mathbb{S}_\gamma(B)$,  $C=A \sqcup B$, $C'=A' \sqcup B'$, $D=A \sqcap B$, and $D'=A' \sqcap B'$. Then, for all $e \in \mathbb{R}$ we have:
\begin{enumerate}
    \item $\frac{\triangle C(e)}{1+\max(\beta, \gamma)} \leq \triangle C'(e) \leq (1+\max(\beta, \gamma)) \triangle C(e)$
    \item $\frac{\triangle D(e)}{(1+\beta)(1+ \gamma)} \leq \triangle D'(e) \leq (1+\beta)(1+\gamma) \triangle D(e)$
\end{enumerate}
\end{lemma}

\begin{proof}
The first claim follows from the assumption that $\oplus$ does not introduce any error and it can be proved by the following:
\begin{align*}
    &(\triangle A(e) \oplus \triangle B(e))/(1+\max(\beta, \gamma))
    \\ 
    \leq& \#C(e) = \triangle A'(e) \oplus \triangle B'(e)
    \\
    \leq& (1+\max(\beta, \gamma))(\triangle A(e) \oplus \triangle B(e))
\end{align*}

The second claim can be also proved similarly; based on definition of $\sqcap$
, we have
\begin{align*}
    \triangle D(e) &= 
    \bigoplus_{j \leq e}\bigoplus_{i \in \mathbb{R}} (\#A (j-i) \otimes \#B (i))
    \\
    &=\bigoplus_{i \in \mathbb{R}} \bigoplus_{j \leq e} (\#A (j-i) \otimes \#B (i))
    \\
    &=\bigoplus_{i \in \mathbb{R}} (\#B(i) \otimes \bigoplus_{j \leq e} \#A (j-i) )
    \\
    &=\bigoplus_{i \in \mathbb{R}} (\#B (i) \otimes \triangle A (e-i))
\end{align*}
Let $D'' = A' \sqcap B$, then based on the approximation guarantee of $\triangle A'(e)$ and the error properties of $\otimes$ and $\oplus$, we have
\begin{align*}
    \triangle D(e)/(1+\beta) \leq \triangle D''(e) \leq (1+\beta) \triangle D(e)
\end{align*}
Then the second claim follows by replacing $B$ with $B'$ in $D''$ and repeat the above step.
\end{proof}


Now we can prove the existence of an algorithm for approximating SumProd FAQ-AI(1) queries. 

\begin{proof}[Proof of Theorem \ref{thm:sumprodmain}]
We first consider the approximation ratio.
Inside-Out on the query $\tilde Q(J)$ performs
the same semiring operations as does on the 
query 
$\widehat{Q}(J)$, but it additionally applies
the  $\alpha$-Sketching operation over each partial results, 
 meaning the algorithm applies $\alpha$-Sketching after steps \ref{step:aggregation},\ref{step:multiplication}, and \ref{step:roottable}. Lets look at each iteration of applying steps \ref{step:aggregation} and \ref{step:multiplication}. Each value produced in the steps \ref{step:aggregation} and \ref{step:roottable} is the result of applying $\bigcirc$ over at most $n^m$ different values (for acyclic queries it is at most $n$). Using Lemma \ref{lemma:epsilonsketch} and the fact that the algorithm applies $\bigcirc$ first on each pair and then recursively on each pair of the results, the total accumulated error produced by each execution of steps \ref{step:aggregation} and \ref{step:roottable} is $m\log(n) \alpha$. Then, since the steps \ref{step:aggregation} and \ref{step:multiplication} will be applied once for each table, and \ref{step:multiplication} accumulates the errors produced for all the tables, the result of the query will be $(m^2\log(n)+m)\alpha$-Sketch of $\widehat{Q}(J)$.

We now turn to bounding the running time of our algorithm.
The time complexity of Inside-Out is $O(m d^2 n^h \log n)$ when the summation and product operators take a constant time~\cite{faq}. The size of each partial result set $A \in S$, after applying $\alpha$-sketching, will depend on the smallest positive value of $\triangle A(e)$ and the largest value of $\triangle A(e)$. Let $\beta$ and $\gamma$ be the minimum and maximum positive real value of SumProd query over all possible sub-matrices of the design matrix, the smallest and largest value of $\triangle A(e)$ for all partial results $A$ would be $\beta$ and $\gamma$ respectively; therefore, the size of the partial results after applying $\alpha$-Sketching is at most $O(\frac{\log(\gamma/\beta)}{\alpha})$. As a result, we compute each of $\bigcirc$ and
$\odot$ in time $O\left(\left(\frac{\log(\gamma/\beta)}{\alpha}\right)^2\right)$. 
Therefore our algorithm runs in time $O(\frac{1}{\epsilon^2}(m^2 \log(n) \log(\frac{\gamma}{\beta}))^2  (d^2 m n^{h} \log(n)))$ and the claim follows by the assumption that the log of the aspect ratio, $\log(\frac{\gamma}{\beta})$, is polynomially bounded.
\end{proof}

\section{Example Applications }
\label{sect:applications}

In this section we give example applications of our results. 

\paragraph{Inequality Row Counting:} Some example problems fow which we can use our Inequality Row Counting to obtain a RAS  in a straightforward manner: 
\begin{itemize}
    \item Counting the number of points on one side of a hyperplane, say the points $x$ satisfy $\beta \cdot x \le L$. 
    \item 
    Counting the number of points within a hypersphere of radius $r$ centered at a point $y$. The additive constraint is $\sum_{i=1}^d (x_i - y_i)^2 \le r^2$.
    \item
    Counting the number of points in a axis parallel ellispoid, say the points $x$ such that $\sum_{i=1}^d \frac{x_i^2}{\alpha_i^2}$ for some $d$ dimensional vector $\alpha$.
\end{itemize}

\paragraph{SumSum FAQ-AI(1) Queries}
Some examples of  problems that can be
reduced to SumSum FAQ-AI(1) queries and 
an application of Theorem \ref{thm:sumsummain} gives a RAS:
\begin{itemize}
    \item Sum of 1-norm distances from a point $y$ of points on one side of a hyperplane. The SumSum query is $\sum_{x \in J} \sum_{i=1}^d |x_i - y_i|$. One can easily verify the addition introduces no error and is repeatable. 
    \item
    Sum of 2-norm squared of points in a axis parallel ellipsoid. The SumSum query is $\sum_{x \in J} \sum_{i=1}^d x_i^2$.  
    \item
    Number of nonzero entries of points on one side of a hyperplane. The SumSum query is $\sum_{x \in J} \sum_{i=1}^d \mathbbm{1}_{x_i \ne 0}$. 
\end{itemize}

\paragraph{SumProd FAQ-AI(1) Queries}
Some examples of  problems that can be
reduced to SumProd FAQ-AI(1) queries and 
an application of Theorem \ref{thm:sumprodmain} gives a RAS:
    \begin{itemize}
        \item
         Finding the minimum 1-norm of any point in a hypersphere $H$ of radius $r$ centered at a point $y$. The SumProd query is $\min_{x \in J} \sum_{i=1}^d |x_i|$. Note $(\mathbb{R}^+ \cup \{0\} \cup \{+ \infty \}, \min, +, +\infty, 0)$ is a commutative semiring. The multiplication operator in this semiring, which is  addition, has bounded error. The addition operator, which is minimum, introduces no error and is monotone. The aspect ratio is at most 
         $(\max_{x \in J} \sum_{i=1}^d |x_i|) / (\min_{x \in J} \min_{ i \in [d] | x_i\neq 0} |x_i|)$, and thus the log of the aspect ratio is polynomially bounded. 
         \item
         Finding the point on the specified side of a hyperplane $H$ that has the maximum 2-norm  distance from a point $y$.
         The SumProd query is $\max_{x \in J} \sum_{i=1}^d (y_i - x_i)^2$. Note that this computes the point with the maximum 2-norm squared distance. One can not directly write a SumProd query to compute the point with the 2-norm distance; We need to appeal to the fact that the closest point is the same under both distance metrics. 
         Note $(\mathbb{R}^+ \cup \{0\} \cup \{- \infty \}, \max, +, -\infty, 0)$ is a commutative semiring. The multiplication operator in this semiring, which is  addition, has bounded error. The addition operator, which is maximum, introduces no error and is monotone. The aspect ratio is at most 
         $(\max_{x \in J} \sum_{i=1}^d |x_i|) / (\min_{x \in J} \min_{ i \in [d] | x_i\neq 0} |x_i|)$, and thus the log of the aspect ratio is polynomially bounded. 
    \end{itemize}

\paragraph{Snake Eyes:} Some  examples of problems for which our results apparently do not apply:

    \begin{itemize}
        \item Finding the minimum distance of any point on a specified side of a specified  hyperplane $H$ to $H$.
        So say the problem is to  find a point $x$ where  $\beta \cdot x \ge L$ and $x \cdot \beta$. 
        The natural SumProd query is $\min_{x \in J}  \sum_{i=1}^d x_i \beta_i$. Note that some of the $x_i \beta_i$ terms maybe be negative, so this doesn't fulfill the condition that the domain has to be over the positive reals. And this appears to be a non-trivial issue basically because good approximations of $s$ and $t$ does not in general allow one to compute a good approximation of $s-t$. We have been call  this the subtraction problem. Using a variation of the proof of Theorem \ref{thm:faq2} one can show that approximating this query to within an $O(1)$ factor is NP-hard. 
        \item Sum of entries of the points lying on one side of a hyperplane. The natural SumSum query is $\sum_{x \in J} \sum_{i=1}^d x_i$.
        Again as some of the  $x_i $ terms may be be negative, we run into the subtraction problem again. 
        \item Aggregate 2-norms of the rows in the design matrix. The natural query is $\sum_{x \in J} \left( \sum_{i=1}^d x_i^2 \right)^{1/2}$, which is neither a SumSum or a SumProd query. 
    \end{itemize}

\section{NP-hardness of FAQ-AI(2) Approximation}
\label{section:approximatehardness}


\begin{theorem}
\label{thm:faq2}
For all $c \ge 1$,    it is NP-hard to $c$-approximate the number of rows in the design matrix  (even for a cross product join) that satisfy  two (linear) additive inequalities. So it is $NP$-hard to $c$-approximate FAQ-AI(2). 
\end{theorem}
\begin{proof}
We reduce from the Partition problem,
where the input is a collection $W = \{w_1, w_2, ...,w_m\}$ of positive integers, and the question is whether one can partition $W$ into two parts with equal aggregate sums.
   From this instance we create $m$ tables, $T_1, T_2, \dots, T_m$, where each $T_i$ has a single columns and has two rows with entries  $w_i$ and $-w_i$. Let 
    $J$ be the cross product of these tables. Note that $J$ has exactly $2^m$ rows and each row $x \in J$ contains either $w_i$ or $-w_i$ for every $i$, which can be naturally interpreted as a partitioning that places each item $i$ in one part or the other, depending on the sign of $w_i$. The two (linear) additive inequalities are $(1, 1, \dots, 1) \cdot x \geq 0$ and $(-1, -1, \dots, -1)\cdot x \geq 0$. 
    Then the solution to the Row Counting SumProd query subject to these two constraints is the number of ways to partition $W$ into two parts of equal aggregate sum. 
\end{proof}

\bibliographystyle{abbrv}
\bibliography{faqai}

\newpage

\appendix

\section{Background}
\label{app:background}

\subsection{Fractional edge cover number and output size bounds}
\label{app:agm}
In what follows, we consider a conjunctive query $Q$ over a relational database instance $I$.
We use $n$ to denote the size of the largest input relation in $Q$.
We also use $Q(I)$ to denote the output and $|Q(I)|$ to denote its size.
We use the query $Q$ and its hypergraph $\calH$ interchangeably.
\bdefn[Fractional edge cover number $\rho^*$]
Let $\calH=(\calV,\calE)$ be a hypergraph (of some query $Q$). Let $B\subseteq\calV$ be any subset
of vertices. 
A {\em fractional edge cover} of $B$ using edges in $\calH$ is a feasible
solution $\vec\lambda =(\lambda_S)_{S\in\calE}$ to the following linear
program:
\begin{eqnarray*}
	\min && \sum_{S\in\calE} \lambda_S\\
	\text{s.t.}&& \sum_{S : v \in S} \lambda_S \geq 1, \ \ \forall v \in B\\
	&& \lambda_S \geq 0, \ \ \forall S\in \calE.
\end{eqnarray*}
The optimal objective value of the above linear program is called
the {\em fractional edge cover number} of $B$ in $\calH$ and is denoted by $\rho^*_\calH(B)$.
When $\calH$ is clear from the context, we drop the subscript $\calH$ and use $\rho^*(B)$.

Given a conjunctive query $Q$, the fractional edge cover number of $Q$ is $\rho^*_\calH(\calV)$
where $\calH=(\calV,\calE)$ is the hypergraph of $Q$.
\edefn

\begin{theorem}[AGM-bound~\cite{AGMBound,GM06}]
Given a full conjunctive query $Q$ over a relational database instance $I$,
the output size is bounded by
\[|Q(I)| \leq n^{\rho^*},\]
where $\rho^*$ is the fractional edge cover number of $Q$.
\label{thm:agm-upperbound}
\end{theorem}

\begin{theorem}[AGM-bound is tight~\cite{AGMBound,GM06}]
Given a full conjunctive query $Q$ and a non-negative number $n$,
there exists a database instance $I$ whose relation sizes are upper-bounded by $n$ and satisfies
\[|Q(I)| =\Theta(n^{\rho^*}).\]
\label{thm:agm-lowerbound}
\end{theorem}

\emph{Worst-case optimal join algorithms}~\cite{LFTJ,Ngo2012wcoj,skew} can be used to answer any full conjunctive query $Q$
in time
\begin{equation}
O(|\calV|\cdot|\calE|\cdot n^{\rho^*}\cdot \log n).
\label{eqn:runtime:lftj}
\end{equation}

\subsection{Tree decompositions, acyclicity, and width parameters}
\label{app:td}
\bdefn[Tree decomposition]
\label{defn:TD}
Let $\calH = (\calV, \calE)$ be a hypergraph.
A {\em tree decomposition} of $\calH$ is a pair $(T, \chi)$
where $T = (V(T), E(T))$ is a tree and $\chi : V(T) \to 2^{\calV}$ assigns to
each node of the tree $T$ a subset of vertices of $\calH$.
The sets $\chi(t)$, $t\in V(T)$, are called the {\em bags} of the 
tree decomposition.  There are two properties the bags must satisfy
\bi
\item[(a)] For any hyperedge $F \in \calE$, there is a bag $\chi(t)$, $t\in
V(T)$, such that $F\subseteq \chi(t)$.
\item[(b)] For any vertex $v \in \calV$, the set 
$\{ t \suchthat t \in V(T), v \in \chi(t) \}$ is not empty and forms a 
connected subtree of $T$.
\ei
\edefn

\bdefn[acyclicity]\label{defn:alpha-acyclic-td}
A hypergraph $\calH = (\calV, \calE)$ is {\em acyclic} iff
there exists a tree decomposition 
$(T, \chi)$ in which every bag $\chi(t)$ is a hyperedge of $\calH$.
\edefn

When $\calH$ represents a join query, the tree $T$ in the above 
definition
is also called the {\em join tree} of the query. 
A query is acyclic if and only if its hypergraph is acyclic.

For non-acyclic queries, we often need a measure of how ``close'' a query is to being acyclic. To that end, we use \emph{width} notions of a query.

\bdefn[$g$-width of a hypergraph: a generic width notion~\cite{adler:dissertation}]
\label{defn:g-width}
Let $\calH=(\calV,\calE)$ be a hypergraph, and
$g : 2^\calV \to \mathbb R^+$ be a function that assigns a non-negative
real number to each subset of $\calV$.
The {\em $g$-width} of a tree decomposition $(T, \chi)$ of $\calH$ is 
$\max_{t\in V(T)} g(\chi(t))$.
The {\em $g$-width of $\calH$} is the {\em minimum} $g$-width
over all tree decompositions of $\calH$.
(Note that the $g$-width of a hypergraph is a {\em Minimax} function.)
\edefn

\bdefn[{\em Treewidth} and {\em fractional hypertree width} are special cases of {\em $g$-width}]
Let $s$ be the following function:
$s(B) = |B|-1$, $\forall V \subseteq \calV$.
Then the {\em treewidth} of a hypergraph $\calH$, denoted by
$\tw(\calH)$, is exactly its $s$-width, and
the {\em fractional hypertree width} of a hypergraph $\calH$,
denoted by $\fhtw(\calH)$, is the $\rho^*$-width of $\calH$.
\edefn

From the above definitions, $\fhtw(\calH)\geq 1$ for any hypergraph $\calH$.
Moreover, $\fhtw(\calH)=1$ if and only if $\calH$ is acyclic.

\subsection{Algebraic Structures}

In this section, we define some of the algebraic structures used in the paper.   First, we discuss the definition of a monoid.  A monoid is a semi-group with an identity element.   Formally, it is the following. 

\begin{definition}
Fix a set $S$ and let $\oplus$ be a binary operator $S \times S \rightarrow S$.  The set $S$ with $\oplus$ is a monoid if (1) the operator satisfies associativity; that is, $ (a \oplus b) \oplus c = a \oplus (b \oplus c)$ for all $a,b,c \in S$ and (2) there is identity element $e \in S$ such that for all $a \in S$, it is the case that $e \oplus a = a \oplus e = e$.

A commutative monoid is a monoid where the operator $\oplus$ is commutative.   That is $a \oplus b = b \oplus a$ for all $a,b \in S$.
\end{definition}

Next, we define a semiring.

\begin{definition}
A semiring is a tuple $(R,\oplus,\otimes,I_0,I_1)$. The $\oplus$ operator is referred to as addition and the $\otimes$ is referred to as multiplication. The elements $I_0$ and $I_1$ are reffered as $0$ element and $1$ element and both are included in $R$. The tuple $(R,\oplus,\otimes,I_0,I_1)$ is a semiring if, 

\begin{enumerate}
\item it is the case that $R$ and $\oplus$ are a commutative monoid with $I_0$ as the identity.
\item $R$ and $\otimes$ is a monoid with identity $I_1$.
\item the multiplication distributes over addition.  That is for all $a,b,c \in R$ it is the case that $a \otimes (b \oplus c) = (a\otimes b) \oplus (a \otimes c)$ and $ (b \oplus c) \otimes a  = (b\otimes a) \oplus (c \otimes a)$.
\item the $I_0$ element annihilates $R$. That is, $a \otimes I_0 = I_0$ and $I_0\otimes a = I_0$ for all $a \in R$.
\end{enumerate}

A commutative  semiring is a semiring where the multiplication is commutative.  That is, $a \otimes b = b \otimes a$ for all $a,b \in S$.
\end{definition}

\end{document}